\documentclass[conference, twocolumn,final]{IEEEtran}
\makeatletter

\usepackage{setspace}
\usepackage[final]{graphicx}
\usepackage[reqno]{amsmath}
\usepackage{amssymb}
\usepackage{cite}
\usepackage{balance}
\usepackage{color}
\usepackage{subfigure}
\usepackage{caption}
\usepackage{amsthm}
\usepackage{url}
\usepackage[inoutnumbered,linesnumbered,algoruled,slide,vlined]{algorithm2e}
\usepackage{algorithmic}

\newfont{\bbb}{msbm10 scaled 500}

\newfont{\bb}{msbm10 scaled 1100}

\newcommand{\Rc}{{\cal R}}
\newcommand{\Xv}{{\bf X}}
\newcommand{\Yv}{{\bf Y}}
\newcommand{\Ck}{{\bf C}}

\newcommand{\argmax}{\operatornamewithlimits{argmax}}
\newcommand{\argmin}{\operatornamewithlimits{argmin}}

\newtheorem{theorem}{Theorem}
\newtheorem{lemma}{Lemma}

    \IEEEoverridecommandlockouts
    
    \author
    {Mehmet Karaca, Saeed Bastani, Basuki Endah Priyanto, Mohammadhassan Safavi  and Bj{\"o}rn Landfeldt }
    \title{Resource Management for  OFDMA based  Next Generation 802.11ax WLANs
    \thanks{This work was sponsored by the European Celtic-Plus project CONVINcE and the EC  FP7  Marie Curie IAPP Project 324515, ``MeshWise"}
    \thanks{
    Mehmet Karaca,  Saeed Bastani,   Mohammadhassan Safavi  and Bj{\"o}rn Landfeldt are with the Department  of Electrical and Information Technologies, Lund University, Lund, Sweden. Email: \{mehmet.karaca,saeed.bastani,mohammadhassan.safavi\}@eit.lth.se, \{bjorn.landfeldt\}@eit.lth.se.}
    \thanks{
    Basuki E. Priyanto is with Sony Mobile Communication AB, Lund,
    Sweden. Email: basuki.priyanto@sonymobile.com.} }

\begin{document}

\maketitle

\begin{abstract}
Recently, IEEE 802.11ax Task Group has adapted  OFDMA as a  new  technique for enabling multi-user transmission. It has   been also decided  that the scheduling duration should be same for all the  users in a multi-user OFDMA so that the transmission  of the users should end at the same time. In order to realize that condition,  the users with insufficient data should transmit null data (i.e. padding) to fill the duration. While this scheme offers strong features such as resilience to Overlapping Basic Service Set (OBSS) interference and ease of synchronization, it also poses major side issues of degraded throughput performance and waste of devices' energy. In this work, for OFDMA based  802.11 WLANs we first propose practical algorithm in which the scheduling duration is fixed and does not change from time to time. In the second algorithm the scheduling duration is dynamically determined in a resource  allocation framework by taking into account the padding overhead, airtime fairness and energy consumption of the users. We  analytically investigate our resource allocation problems  through  Lyapunov optimization techniques and  show that our algorithms are arbitrarily close to the optimal performance at the price of reduced convergence rate.  We also calculate  the overhead of our algorithms in a realistic set-up and propose solutions for the implementation issues.
\end{abstract}

\section{Introduction}
The increasing demand for high throughput wireless access is driven by the proliferation of mobile devices, an increasing demand for data-hungry services, and the growing trend of dense network scenarios. This has led to an unprecedented growth in the wireless local area network (WLAN) market, which has spurred a new wave of standardization activities, leading to the recently developed multi-gigabit IEEE 802.11ac  \cite{ieeeac} followed by IEEE 802.11ax ( (High Efficiency WLAN (HEW)) \cite{ieeeax} (see TGax Specification Framework) effort, with an ambitious target of achieving at least a four times increase of medium access control (MAC) throughput per station compared to 802.11ac. This target will be far-reached unless radical improvements are made in both physical layer  as well as medium access control functionality. The mature experience of using 802.11 based WLANs indicates that a simple CSMA/CA mechanism is inefficient especially when the network density or the traffic volume increase \cite{bianchi2000performance}, despite the fact that both cases will be inevitable properties of future WLAN deployment scenarios.     

Orthogonal frequency division multiple access (OFDMA), as a multi-user (MU) transmission technique, has been approved as new technique for 802.11ax to improve the performance of dense 802.11 networks by offering multi-user diversity and a high spectral efficiency, thus providing substantially enhanced throughput and paving the way for the realization of multi-gigabit WLANs \cite{ieeeax}. 
However, attaining the full advantages of OFDMA, especially in dense network scenarios requires that certain features are taken into consideration including low synchronization complexity and high resilience to interference from OBSS. These features are indeed addressed in the recent specification of IEEE 802.11ax standard \cite{IEEEOFDMAFW} by proposing that a scheduling duration (e.g.,  Physical Layer Convergence Protocol (PLCP) Protocol Data Unit (PPDU) or  Transmit Opportunity (TXOP) duration) should be announced to the users so that  each user in a MU-OFDMA should end their transmission at the same time. However,  the determination of scheduling duration has not been specified yet. Also,  when the users do not have sufficient data to transmit 802.11ax mandates that the scheduled users to transmit null bits (i.e. padding) to fill this duration, which degrades throughput performance and waste of devices' energy caused by the padding bit transmissions. We note that padding overhead also occurs in other MU transmission techniques such as MU-MIMO in 802.11ac \cite{Gond11ac}, where MAC padding bits should be introduced as well.


In this paper, we first propose two algorithms: in the first algorithm, the scheduling duration is fixed at every transmission time. On the other hand, in the second algorithm we allow that the scheduling duration can change over time depending on queue sizes and channel condition of  users.  Then, we develop resource allocation policies for the second algorithm in which we optimally determine the scheduling duration  for the minimization of the padding overhead by taking into account airtime fairness and energy consumption of users. We also calculate how much overhead these policies will have in practice, and point out their implementation issues.      
\section{Related Work}\label{background}
OFDMA  has already been used in recent technologies including LTE \cite{LTE}. However, its applications to 802.11-based WLANs is currently under investigation. 
In a majority of studies, including \cite{fallah2008hybrid,kwon2009generalized,wang2011two,jung2012group,lou2014multi}, enhancing the throughput has been the main focus. \cite{fallah2008hybrid} proposed a hybrid CSMA/CA and OFDMA MAC protocol operating in separate phases of transmission request and data transmission. In the request phase, the CSMA/CA mechanism is used, with the nodes contending in separate sub-channels allocated by the access point (AP). In a second phase, the AP schedules the winning stations in the time domain, and sends a scheduling frame to the nodes. Thus, OFDMA is only used during the transmission request phase. In contrast, in this paper we exploit OFDMA for data transmission to achieve transmission diversity. Kwon et al. \cite{kwon2009generalized} proposed another hybrid OFDMA and CSMA/CA MAC protocol which differs from \cite{fallah2008hybrid} in two ways: first, it exploits OFDMA for resource allocation in the  transmission phase. Second, no dedicated sub-channels are used in the contention phase. Wang et al. \cite{wang2011two} proposed to trigger multiple contention threads in each node. In this approach, a node can simultaneously contend for all sub-channels and use as many sub-channels as it seizes. Both uplink and downlink traffic where addressed in \cite{wang2011two} compared to \cite{fallah2008hybrid,kwon2009generalized} where only uplink traffic was addressed. In another work, Jung et. al. \cite{jung2012group} proposed GC-OFDMA, a group based hybrid OFDMA and CSMA/CA MAC protocol. Like \cite{fallah2008hybrid}, GC-OFDMA allows STAs to use CSMA/CA to contend on separate sub-channels. However, it differs from \cite{fallah2008hybrid} in that the STAs are organized in different groups, and only STAs in the same group contend for transmission request, though on different sub-channels. Second, GC-OFDMA applies multiuser resource allocation in the data transmission phase. 


%
%
In contrast with the new OFDMA scheme proposed by IEEE 802.11ax Task Group \cite{ieeeax},  the aforementioned research studies do not take into account the padding overhead whose effect can be significant without any good resource allocation. Moreover, to the best of our knowledge  the  literature  is  largely  silent  on  the  analysis
of the effect of this overhead. In this paper, unlike previous studies, our work addresses  new problems arising from padding overhead in MU transmission by particularly focusing on the new OFDMA implementation in IEEE 802.11ax.        
\section{System Model and Research Problem}
We consider a fully-connected WLAN topology where in total there are
$N$ users with  uplink\footnote{ Although the problem that we consider and the proposed algorithms are applicable to both downlink and uplink scenarios, we consider only uplink scheduling since it is more challenging than the downlink scheduling due to the necessity of the AP acquiring queue size information from the STAs.} data.  We assume random channel gains between the AP  and the STAs that are independent across time and STAs. In practice, there is only a discrete finite set of $M$ Modulation and coding Schemes (MCS)  available, only a fixed set of data rates
$\Rc=\{r_1,r_2,\ldots,r_M\}$ can be supported.

In our model, we adopt OFDMA for uplink transmission in which the simultaneous transmission of multiple users with lower data rates is possible
by dividing the available bandwidth  into many sub-bandwidth (i.e., sub-channels).  We assume that $K$ users can be scheduled due to the total bandwidth limitation, and $K \leq N$. As an example, if the AP sets the channel bandwidth to 20 MHz for the current transmission, and the bandwidth of each sub-channel is 5 MHz then at most 4 users can be scheduled at that transmission time (i.e., $K=4$).

We consider a group based transmission where  there are $L$  groups with  $K$ users\footnote{With this method, some users may not be assigned to any group since $K$ must be integer value. To be fair, those users can be replaced with the users which are already assigned to a group in an ordered manner.}  and each user  is assigned to one of these groups. By taking into account the practical limitation, which we will explain later, we consider round-robin type scheduling algorithm that is simple to implement, and distributes the total resources evenly among the groups. In round-robin scheduling a group is scheduled at each scheduling time (e.g., every 10 ms)  in a  cycle order for uplink transmission. Then the users in the scheduled group transmit their data to  the AP. We use the indicator
variable $I^{g}(t)$, and
$I^{g}(t)=1$ if group  $g$  is scheduled for transmission in slot $t$,
and $I^{g}(t)=0$ otherwise, and $g\in \{1,2,\cdots,G\}$. Each user maintains a separate queue. Packets arrive
according a stationary arrival process that is independent across
users and time slots. Let $A_{k}^g(t)$ be the amount of data arriving
into the queue of user $k$ in group $g$ at time slot $t$, where $k\in \{1,2,\cdots,K\}$ and $g\in \{1,2,\cdots,G\}$. Let $Q_k^g(t)$ and
$R_k^g(t) \in \Rc$ denote the queue length and transmission rate of user $k$ in group $g$
at time $t$, respectively.
The queue length dynamics for user $k\in \{1,2,\cdots,K\}$ in group  $g\in \{1,2,\cdots,G\}$  are given as
\begin{align}
Q_k^g(t+1)=[Q_k^g(t)-R_k^g(t)T_s(t)I^{g}(t)]^+ + A_k^g(t),   \label{eq:queuelength}
\end{align}
where $[y]^+=\max(y,0)$, $T_s(t)$ is the length of  the \textit{scheduling duration}  at time $t$. In other words,  the transmission of any user cannot take longer than $T_s(t)$ seconds. For analytical simplicity, we assume that a user only transmits data information, and the MAC and PHY layers overhead are neglected. However we take into account these overhead in Sec IV-C. Let $T_k^g(t)$ be the time that is required for the user $k$ to transmit all its data (i.e., $Q_k^g(t)$)  to the AP at transmission time $t$.  Then, we have,
\begin{align}
T_k^g(t)= \frac{Q_k^g(t)}{R_k^g(t)}. \label{eq:Tk}
\end{align}
We note that our analysis and algorithms are group based, and since each group is independent of other groups, and we omit the group index $g$ in the rest of the paper. 

Fig. 1 depicts a reference model of  OFDMA uplink (UL) scenario for an arbitrary group  with  $K$ scheduled users  at a time, which is similar to the one proposed in 802.11 ax WG \cite{IEEEOFDMA}. We call this method  as \textbf{Fixed PPDU (F-PPDU) algorithm}, where a Trigger Frame (TF) which indicates the group ID (or users IDs) that will be scheduled at that time, and also indicates the sub-channels that are assigned to each user in the group is sent by the AP.  TF also announces the scheduling duration $T_s$ which is fixed for each group at every scheduling time. Then, the user in the scheduled group transmit their uplink data to the AP. Note that $T_s$ can be too long for some users with short transmission duration due to small queue sizes and/or high transmission rate. Also, it may not be sufficient for the users with high arrival rate and/or bad channel conditions. As an example, for 802.11ac, the maximum transmission time is set to 5.484 ms regardless of  MCS level\cite{ieeeac}.

The main problem with the reference model is that fixing scheduling duration for each group and at every time can cause significant channel under utilization (e.g., queue sizes of users are small) or communication interruption (e.g.,  delay due to the insufficient amount of time to send necessary amount of data).  Moreover, the users with the transmission duration that is less than $T_s$ must send padding bits  until the end of the scheduling time.  Basically there are two  main reasons behind transmitting padding bits \cite{IEEEOFDMA},\cite{IEEEOFDMAFW}: first,  users can be able to complete their transmission at the same time and, any synchronization issues can be easily mitigated, and sending Block acknowledgment (BA) can be implemented in practice. Second,  if a user completes its transmission before the scheduling duration expires, and becomes silent (i.e., go to sleep mode) any other users from other WLANs (i.e., Overlapping Base Station Subsystem (OBSS ) can sense the channel idle, and start their transmission which can collide with the ACK (acknowledgment) packet transmission of the user. Due to these two important problems, it is mandatory to transmit padding bits in OFDMA transmission for 802.11ax \cite{IEEEOFDMAFW}. In practice, the padding bits are  is a sequence of 0 bits, and do not contain any information. Therefore, sending these bits wastes capacity and causes STAs to expend valuable energy. It is  also possible to use sub-channels with different number of sub-carriers instead of using a fixed number to align with the scheduling time. However, this method causes a user to achieve lower throughput, and also since only limited number of sub-channel width is available, there is always a certain amount of padding overhead in practice \cite{IEEEOFDMAFW}. Even so, the optimization of the scheduling duration with different size of sub-channels may be of independent interest.

We note that in F-PPDU algorithm, the AP is not aware of the  queue sizes, and channel conditions of the uplink users. Hence, the performance of F-PPDU can be very poor in practice due to the fixed $T_s$ and padding overhead. Next, we propose a new protocol to acquire these information, and to optimize the scheduling duration dynamically by taking into account user's requirements and energy consumption.
\section{MU-UL OFDMA with Buffer Status} 
\begin{figure}[t]
    \centering
    \includegraphics[width=0.42\columnwidth]{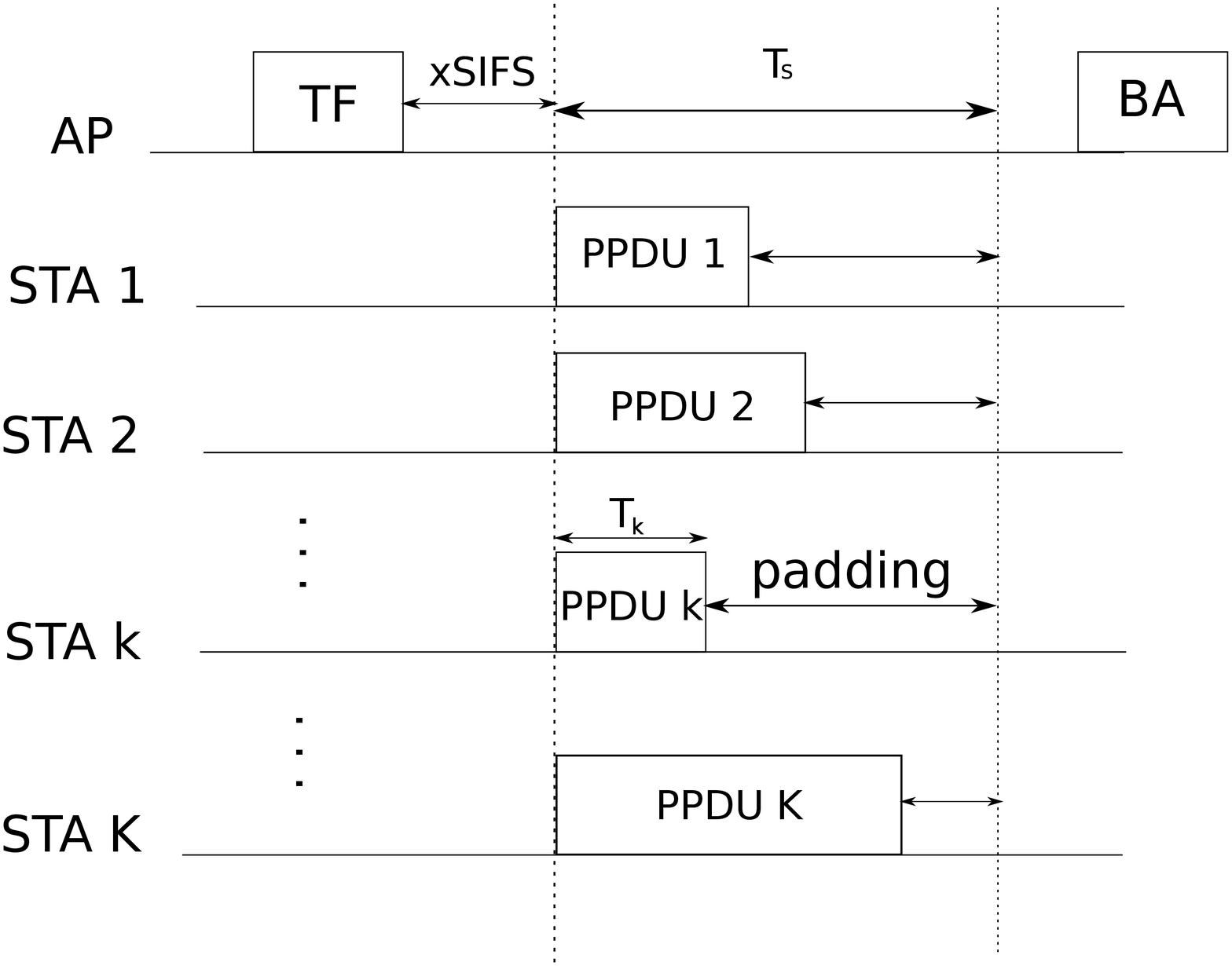}
    \caption{UL OFDMA  with fixed scheduling duration, $T_s$ }
    \label{fig:th}
\end{figure}
In order to optimize the scheduling duration, the AP has to acquire the queue backlog information from the scheduled users at each time $t$. To enable that  we modify F-PPDU algorithm, and present a new protocol namely \textbf{Dynamic-PPDU (D-PPDU)} as follow: the AP first sends a TF which now only indicates the ID of the group (i.e, the IDs of the scheduled users), and the dedicated sub-channel to each scheduled user. After receiving the TF,  each user sends its buffer status (BS) information within the BS frame (i.e., $Q_k(t)$) by using its dedicated sub-channel. Also, the AP uses the BS frame to estimate the uplink transmission rate of user $k$ (i.e., $R_k(t)$). Finally, the AP has both queue size and channel state information of each scheduled user at that time, and can determine the optimal scheduling duration denoted by $T_s^*(t)$. After optimizing the scheduling duration, the AP sends the optimal scheduling duration within  Optimal scheduling Time (OT) frame. Then, each user determines how much data it should transmit, and adjusts its transmission duration which must be shorter or equal than $T_s^*(t)$.  Clearly, this new protocol introduces a certain level of overhead due to the frame exchanges, which we quantify at the end of this section. Next, we calculate the optimal scheduling duration that maximizes the total throughput for a given arbitrary group at a scheduling time $t$. 
\begin{figure}[t]
    \centering
    \includegraphics[width=0.54\columnwidth]{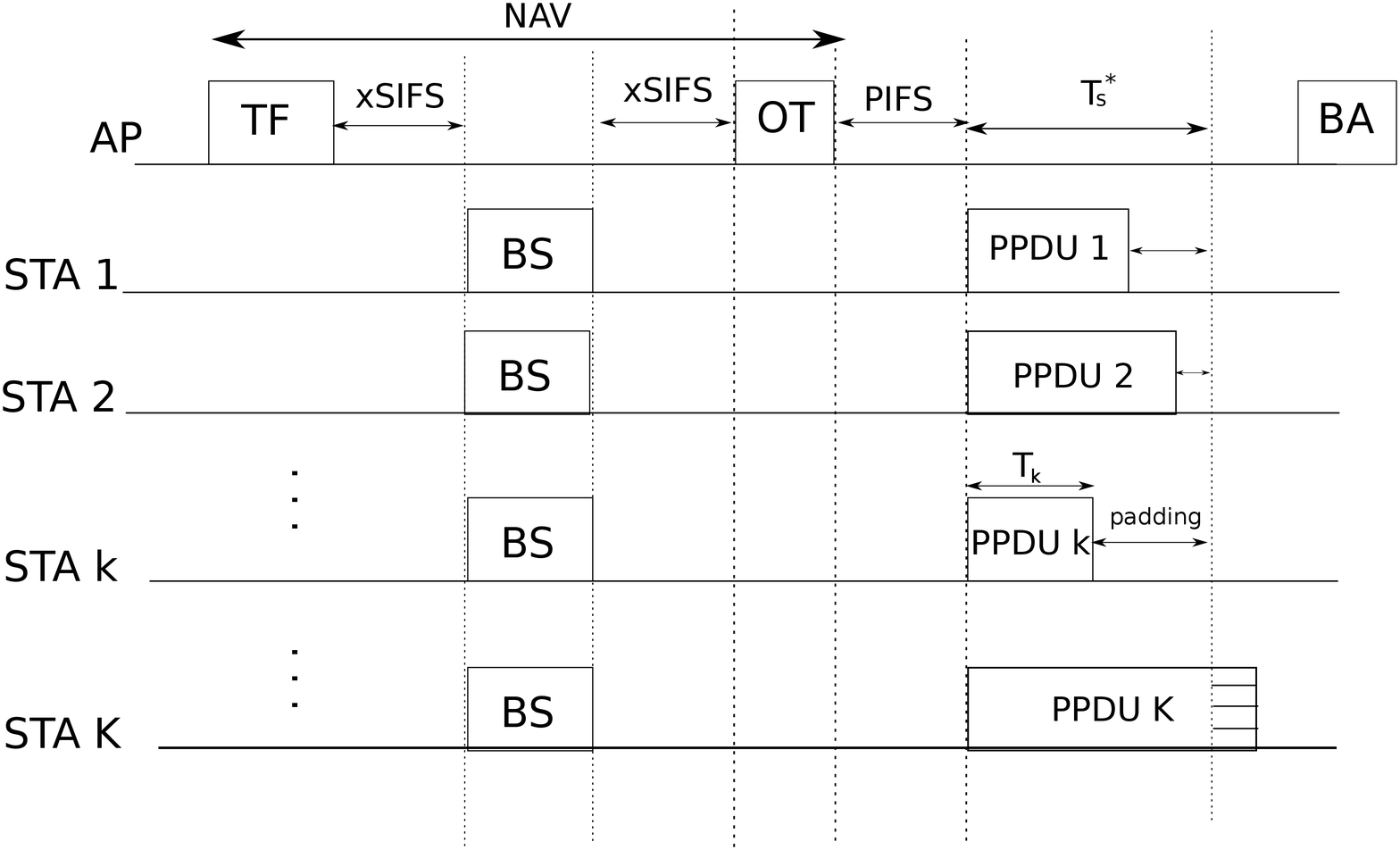}
    \caption{UL OFDMA with optimized scheduling duration }
    \label{fig:th}
\end{figure} 
For a given $T_s(t)$, some users can transmit only a portion of its buffer due to large queue sizes and/or low transmission rate whereas some user can transmit all the data in their buffer. Let $q_{k}(t)$ be the amount of data in bits that is transmitted by user $k$, which depends on  $T_s(t)$, the current buffer size and transmission rate. Then, the total throughput is determined as follows: 
\begin{align}
D_{tot}(t)= \frac{q_1(t) + q_2(t)+ \dots, q_K(t)}{T_{s}(t)}=\sum_{k=1}^K \frac{q_k(t)}{T_{s}(t)} \label{eq:dtot}
\end{align}
Let $T_{min}(t)=\min_{k} {T_k(t)}$ and $T_{max}(t)=\max_{k} {T_k(t)}$, where $T_k(t)$ is given in \eqref{eq:Tk}. By assuming that $T_{min}(t) \leq T_s(t) \leq T_{max}(t)$ for all $t$, then the following Lemma determines the optimal scheduling  duration (e.g. PPDU duration) denoted as $T_s^*(t)$ that maximizes the total throughput at time $t$.
\begin{lemma}
\label{thm:main} For given $Q_k(t)$ and $R_k(t)$ for each user $k \in \{1,2, \cdots, K\}$ at time $t$, the optimal $T_s^*(t)$ that maximizes $D_{tot}(t)$ is equal to $T_{min}(t)$. i.e., $T_s^*(t)=T_{min}(t)$.
\end{lemma}
\begin{proof}
First we determine the throughput of a scheduled user $k$ denoted by $D_k(t)$ for a given $T_s(t)$.If $T_s(t) > T_k(t)$, then all the data in the buffer of user $k$ is transmitted (i.e., $q_k(t)=Q_k(t)$), and $D_k(t)= \frac{Q_k(t)}{T_s(t)}$. If $T_s(t) < T_k(t)$, then only $q_k(t)=Q_k \frac{T_s(t)}{T_k(t)}$ amount of bits\footnote{ In practice only discrete number of bits can be transmitted, however, for simplicity we neglect this fact.} can be transmitted within $T_s(t)$ seconds, hence,  by using \eqref{eq:Tk} $D_k(t)$ is equal to the transmission rate, i.e., $D_k(t)=R_k(t)$. If $T_s(t) = T_k(t)$ user $k$ can discharges its queue, and  $D_k(t)=R_k(t)$.

First, let us assume that $T_s(t) =T_{max}(t)$, then  all the scheduled users would be able to empty their queues. Thus, the overall throughput is given by,
 \begin{align*}
D^1_{tot}(t)= \frac{Q_1(t) + Q_2(t)+ \dots, Q_K(t)}{T_{max}(t)}=\sum_{k=1}^K \frac{Q_k(t)}{T_{max}(t)} 
 \end{align*}
Now, let us assume that $T_s(t) =T_{min}(t)$,  then a scheduled user $k$ can transmit only $Q_k(t) \frac{T_{min}(t)}{T_k(t)}$ amount of bits, and  the throughput of that user will be equal to $ Q_k(t) \frac{T_{min}(t)}{T_k(t)T_{min}(t)}= \frac{Q_k(t)}{T_k(t)}$. In this case, the total throughput becomes
 \begin{align*}
D^2_{tot}(t)= \frac{Q_1(t)}{T_1(t)} + \frac{Q_2(t)}{T_2(t)} + \dots, + \frac{Q_K(t)}{T_K(t)}= \sum_{k=1}^K \frac{Q_k(t)}{T_{k}(t)}
 \end{align*}
Clearly, $D^2_{tot}(t) > D^1_{tot}(t)$. Finally, let us assume $T_{min}(t) < T_s(t)  < T_{max}(t)$. In this case, some of the scheduled users can empty their queues, some of them can only transmit a portion of the data in their queues. Without loss of generality, let $J_1=\{Q_{a_1}(t), Q_{a_2}(t),\dots, Q_{a_A}(t)  \}  $ be the group of users which can discharge their queues, and let $J_2=\{Q_{b_1}(t), Q_{b_2}(t),\dots, Q_{b_B}(t)  \}  $ be the group of users which cannot empty their queues. Then, the total throughput of group $J_1$ and group $J_2$ are given by $D_{J_1}(t)=\frac{Q_{a_1}(t)+  Q_{a_2}(t) +\dots + Q_{a_A}(t) }{T_s(t)}$, and  $D_{J_2}(t)= \frac{Q_{b_1}(t)}{T_{b_1}(t)} + \frac{Q_{b_2}(t)}{T_{b_2}(t)}+ \dots + \frac{Q_{b_B}(t)}{T_{b_B}(t)} $, respectively. The overall throughput is equal to $D^3_{tot}(t)= D_{J_1}(t) + D_{J_2}(t)$.
Clearly, $ D^2_{tot}(t) > D^3_{tot}(t) > D^1_{tot}(t)$. Hence, the maximum overall throughput is achieved when $T_s^*(t) = T_{min}(t)$. This completes the proof.
\end{proof}
\subsection{Minimizing Padding  Overhead}
We note that if the AP sets $T_s^*(t) = T_{min}(t)$ at every time, only the users with  low queue size and/or good channel conditions can empty their queue whereas the other users have to send their remaining data in the next scheduling time. This setting has two drawbacks: it causes unfairness in using airtime and  some users experience longer delay. These drawbacks can be eliminated by setting $T_s(t)$ to large values for instance $T_s^*(t) = T_{max}(t)$. However, at this time we sacrifice the throughput since  more padding overhead must be transmitted. Here, we are interested in the problem of minimizing the padding overhead while resolving the fairness and delay issues. 

Let $H_k(t)$ be the padding overhead of user $k$  in an arbitrary group at time $t$:
\begin{align}
H_k(t) =& \left\{ \begin{array}{l l}
                     T_s(t) - T_k(t)          & \text{; if $T_s(t) > T_{k}(t)$}\\
                     0      & \text{; otherwise}
                \end{array} \label{eq:po}
    \right.
\end{align}
Let $H_{tot}(t)$ be the total padding overhead which is given as $H_{tot}(t) = \sum_{k=1}^{K} H_{k}(t)$.  To overcome the fairness issue we first define our fairness  parameter which measures the proportion of the time at which user $k$ can empty its buffer (i.e., all data in the buffer is transmitted) as follows: 
\begin{align}
F_k(t) =& \left\{ \begin{array}{l l}
                    1            & \text{; if $T_s(t) \geq T_k(t)$}\\
                    0      & \text{; otherwise}
                \end{array} \label{eq:U}
    \right.
\end{align}
In other words, if  $T_s(t) \geq T_k(t)$ then user $k$ is able to transmit all its data at time $t$. If $T_s(t) < T_k(t)$, then user $k$ cannot empty its buffer at that time.  We
define the time average total expected padding overhead for a given group as follows.
\begin{align}
\bar{H}_{tot} = \limsup_{t\rightarrow\infty}
\frac{1}{t}\sum_{\tau=0}^{t-1}{\mathbb E}[{H}_{tot}(\tau)],
\end{align}
where the expectation is taken over the random transmission rates (random
arrival and channel conditions). Then we consider the following optimization problem which aims to minimize the total expected padding  overhead while satisfying  the minimum performance
constraint in terms of the defined fairness denoted as $C_k$  for each user $k$:
\begin{align}
\min\limits_{T_s(t)}\ & \bar{H}_{tot} \notag\\
\text{s.t.}\ {\mathbb E}&\left[F_k(t)\right]
\geq C_k, \quad k\in\{1,2,\ldots,K\}
\label{eq:problem}
\end{align}
where $0 < C_k \leq 1$, and the long-run fraction of time that users complete transmitting all data should be at least $C_k$. As an example if $C_k=0.6$ for user $k$,  it means that 60\% percent of the scheduled times, user $k$ requires to be able to empty its buffer.   It is worth noting that the scheduling duration is optimized for a given number of users due to the round-robin policy. One can also optimize the users to be scheduled. However, it requires the global queue size information from all users in the network, which may not be unaffordable in practice. Due to this reason and its simplicity we employ round-robin policy where only queue size information of a group is needed.  
We next present our solution to the optimization problem defined  in (\ref{eq:problem}). We recall that the dynamics in each groups evolve independently from other groups. Then,  the optimal strategy to solve problem (\ref{eq:problem})  is the same for each group. 

Since the problem in (\ref{eq:problem}) is a stochastic optimization problem with the expected objective and constraints one can use the techniques in \cite{Liu:Comp03}. In this work, we  solve problem
\eqref{eq:problem} using the stochastic network optimization tool of
\cite{Georgiadis:Resource06} which can provide more simplistic solutions \cite{Karaca:cognitive}. First, for each of the
constraints in \eqref{eq:problem}, we construct a virtual queue such
that the queue dynamics for user $k$ in the group is given by
\begin{align}
X_k(t+1) &= [X_k(t) - F_k(t)]^+ + C_k, \label{eq:X}
\end{align}
$k\in\{1,2,\ldots,K\}$, where $[x]^+ \triangleq \max\{x,0\}$. Note
that stabilizing the queues in \eqref{eq:X} is equivalent to
satisfying the constraints in \eqref{eq:problem} since a queue is
stable if the arrival rate is less than the service rate. Let $\Xv(t) = (X_1(t) X_2(t) \cdots
X_K(t))$ be the vector of virtual queues of $K$ users.  In practice, it is upper-bounded by a constant value (i.e., $T_s(t) \leq T_s^{max}$ for all $t$, and  the maximum PPDU duration is set to 5.484 ms with 802.11ac). Hence, $H_{tot}(t)$ is bounded such that  $H_{tot}(t) \leq H_{max}$  for all $g\in\{1,2,\ldots,G\}$.  We define the following
quadratic Lyapunov function and conditional Lyapunov drift:
\begin{align}
&L(\Xv(t)) \triangleq \frac{1}{2} \sum_{k=1}^{K} X_k^2(t),\\
&\Delta(\Xv(t)) \triangleq {\mathbb E}\left[ L(\Xv(t+1)) -
L(\Xv(t)) | \Xv(t)\right]\label{eq:drift1}.
\end{align}
The following Lemma is useful in establishing the optimality of our
algorithm.
\begin{lemma}
\label{lemma:1} For every time slot $t$ and any policy that determines $T_s(t)$, the
following bound holds:
\begin{align}
\Delta&   (\Xv(t)) + V{\mathbb E}[H_{tot}(t)|\Xv(t)] \leq B_1 +
\sum_{k=1}^{K}X_k(t)C_k\notag\\& + 
\sum_{k=1}^{K} {\mathbb E}[H_k(t)V-X_k(t)F_k(t) |\Xv(t) ]  \label{eq:lemma1}
\end{align}
where $B_1 = \frac{1}{2}\left(\sum_{k=1}^{K}C_k^2 +
K\right)$ and $V$ is a system parameter that
characterizes a tradeoff between performance optimization and delay
in the virtual queues.
\end{lemma}
\begin{proof}
We can write the following inequality by using the fact $([a]^+)^2
\leq (a)^2, \quad \forall a$:
\begin{align*}
X_k^2(t+1) \leq& \ X_k^2(t) + C_k^2 + (1)^2 -
2X_k(t)[F_k(t)-C_k]
\end{align*}
for $k\in\{1,2,\ldots,k\}$.  Therefore,
the Lyapunov drift in \eqref{eq:drift1} can be upper bounded as
\begin{align}
\Delta_1(\Xv(t)) \leq& \ B_1 - \sum_{k=1}^{K} {\mathbb E} [X_k(t)(F_k(t)- C_k) | \Xv(t)] \label{eq:Lyp_drift}
\end{align}
where $B_1 = \frac{1}{2}\left(\sum_{k=1}^{K}C_k^2 +
K\right)$. In addition, we define a cost function
$\mathbb{E}[H_{tot}(t)|\Xv(t)]$ as the expected total padding overhead during time slot $t$. After adding the cost function
multiplied by $V$ to both sides of (\ref{eq:Lyp_drift}),  and expanding and rearranging
the terms follows \eqref{eq:lemma1}.
\end{proof}
Now, we present our D-PPDU algorithm which  optimally solves  the problem in \eqref{eq:problem}.

\textbf{D-PPDU Algorithm:}
At  time slot $t$ suppose that group $g$ is scheduled according to the round-robin policy. Then, observe the virtual queue backlog $X_k(t)$ for each user  $k$ in the group and the transmission duration $T_k(t)$, and determine $T_s^*(t)$ solving the following
minimization problem:
\begin{align}
& T_s^*(t) =  \label{eq:policy1} \\ \notag
 &\argmin\limits_{T_{min}(t) \leq T_s(t) \leq T_{max}(t)}\left\{\sum_{k=1}^{K} H_k(t) -\frac{X_k(t)}{V}F_k(t)\right\}. 
\end{align}
Then, update the virtual queues according to the queue dynamics in
\eqref{eq:X}.
\label{thm:1}

It is worth noting that parameter $V$ specifies a tradeoff between optimality and the
average length of the virtual queues. Thus, for large virtual
queues, the system experiences larger transient times to achieve the
optimal performance and hence needs more time to adapt to possible
changes in channel and arrival statistics \cite{Georgiadis:Resource06},\cite{Karaca:cognitive}. The structure of D-PPDU algorithm in \eqref{eq:policy1}  suggests that if $T_s^*(t)$ is  unfavorable for a user $k$, then, consequently, the virtual queue size of that user will increase. In this case,  at the next scheduling time the resulting optimal $T_s^*(t)$ will be larger to satisfy that user  according to the policy in \eqref{eq:policy1}. 

It is important to note that depending on the distribution of transmission duration it may not be always possible to satisfy the set of the constraints of the users.  Clearly the feasibility depends on $T_s^{max}$.  We here assume that $C_k$, $k \in \{1,2,\dots,K\}$ are feasible, and let the feasibility region
of problem \eqref{eq:problem} be $\Lambda$ and let {\boldmath $\epsilon$}
$\triangleq (\epsilon\ \epsilon\cdots\ \epsilon)$. Note that If the vector $\Ck = (C_1\ C_2\cdots C_K)$ is
feasible (i.e., $\Ck \in \Lambda$), then there exists $\epsilon > 0$
such that ($\Ck + ${\boldmath $\epsilon$}) $\in \Lambda$.

Next, we give our performance bound of D-PPDU algorithm. 
\begin{theorem} 
Let $H_{tot}^{*}$ be the minimum  padding overhead for problem (\ref{eq:problem}). If  $\Ck$
is strictly interior to $\Lambda$, then D-PPDU satisfies the following bound:
\begin{align}
\Delta (t) +  V \mathbb{E}[H_{tot}|\Xv(t)] \leq B_1-\epsilon
\sum_{i=1}^N X_k(t) + V H_{tot}^{*} , \label{eq:boundth1}
\end{align}
where $ \Ck + ${\boldmath $\epsilon$} $\in \Lambda$. Then, applying the theorem in  \cite{Georgiadis:Resource06},  the system is stable and the time average backlog satisfies:
\begin{align}
 \limsup_{T\rightarrow \infty} \frac{1}{T} \sum_{t=0}^{T-1}
\sum_{k=1}^K \mathbb{E} [X_k(t)] \leq \frac{B_1+V H_{tot}^{*}}{\epsilon}  \label{eq:virbound}
\end{align}
and
\begin{align}
 \limsup_{T\rightarrow \infty} \frac{1}{T} \sum_{t=0}^{T-1}
 \mathbb{E}[H_{tot}] \leq H_{tot}^{*} +  \frac{B_1}{V}
 \label{eq:optbound}
\end{align}
\end{theorem}
\begin{proof}
The proof is similar to the proof of Theorem in \cite{Georgiadis:Resource06}. The proof involves showing that there exits  a randomized algorithm that achieves the minimum padding overhead for a given $ ${\boldmath $\epsilon$}, and  satisfies the constraint without taking into account the virtual queue sizes.  Next,  since D-PPDU minimizes the right-hand side of the bound \eqref{eq:lemma1} all the time, then the proof shows that  D-PPDU satisfies  \eqref{eq:boundth1}. Then by following the similar techniques to those in \cite{Georgiadis:Resource06}, the bounds \eqref{eq:virbound} and \eqref{eq:optbound} are found.
\end{proof} 
\subsection{ Energy Consideration }
\label{sec:energy} Energy consumption of  wireless devices is an essential problem especially for uplink users with limited energy resources. However, at the same time  uplink users aim to transmit as much data as possible to empty its buffer since usually the uplink traffic is low, and buffering data in the next transmission time can lead large delay. For this purpose, our
objective is to investigate  the policies that can capture a good tradeoff between the energy consumption for transmitting padding bits and the average number of time at which a user can discharger its buffer to avoid large delay. 
Our objective is to find an optimal policy that can determine  scheduling duration dynamically by taking account this tradeoff. For given $T_s(t)$ we define the energy consumption of user $k$ caused by transmitting both data and padding bits is equal to $E_k(t) =T_s(t)  \times P$. 
Since $T_s(t)$ and $P$  are upper bounded, $E_k(t) \leq E_{max}$ for each user at every scheduling time. By using the definition of $F_k(t)$ in \eqref{eq:U}, we define ${S}_{tot}(t)$ as ${S}_{tot}(t) =  \sum_{k=1}^{K} F_k(t)$
and ${S}_{tot}(t) \leq K$. The time average expected ${S}_{tot}(t)$  is given as:
\begin{align}
\bar{S}_{tot} = \liminf_{t\rightarrow\infty}
\frac{1}{t}\sum_{\tau=0}^{t-1}{\mathbb E}[{S}_{tot}(\tau)],
\end{align}
where the
expectation is taken over the random transmission rates (random
arrival and channel conditions). Then, we consider the following optimization problem:
\begin{align}
\max\limits_{T_s(t)}\ & \quad \bar{S}_{tot} \notag\\
\text{s.t.} \ {\mathbb E}&\left[E_k(t)\right]
\leq E_k^{tot}, \quad k\in\{1,2,\ldots,K\} \label{eq:problem2}
\end{align}
where $ E_k^{tot}$ is the average energy consumption that a user $k$ can spend during the overall transmission including the transmission of padding bits, and $E_{max}^{tot}=\max_k E_k^{tot} $.  We construct a virtual queue such
that the queue dynamics for user $k$ in an arbitrary group is given by,
\begin{align}
Y_k(t+1) &= [Y_k(t) - E_k^{tot}]^+ + E_k(t), \label{eq:Y}
\end{align}
where $k\in\{1,2,\ldots,K\}$, and let $\Yv(t) = (Y_1(t) Y_2(t) \cdots
Y_K(t))$ be the vector of virtual queues.
By following the same approach that we use for the derivation of  D-PPDU the algorithm, we next give our Energy Aware D-PPDU (EAD-PPDU) algorithm.

\textbf{EAD-PPDU Algorithm:}
At  time slot $t$ suppose that group $g$ is scheduled according to the round-robin policy. Then, observe the virtual queue backlog $Y_k(t)$ for each user  $k$ in the group  and the transmission duration $T_k(t)$, and choose $T_s^*(t)$ solving the following
maximization problem:
\begin{align}
 T_s^*(t) =  \label{eq:policy2}  
 \argmax_{T_s(t)}\left\{\sum_{k=1}^{K} F_k(t) -\frac{Y_k(t)}{V}E_k(t)\right\}. 
\end{align}
\label{thm:2}
We note that for problem (\ref{eq:problem2}) $T_s(t)$  can be less than  $T_{min}(t)$ due to the total energy constraint.  Next, we give our performance bound of EAD-PPDU algorithm. 
\begin{theorem} 
Let $S_{tot}^{*}$ be the optimal solution for problem (\ref{eq:problem2}). Then EAD-PPDU satisfies the following bound:
\begin{align}
\Delta (t) -  V \mathbb{E}[S_{tot}(t)|\mathbf{Y}(t)] \leq B_2-\epsilon
\sum_{k=1}^K Y_k(t) - V S_{tot}^{*} , \label{eq:boundth2}
\end{align}
Then, applying the theorem in  \cite{Georgiadis:Resource06},  the system is stable and the time average backlog satisfies:
\begin{align}
 \limsup_{T\rightarrow \infty} \frac{1}{T} \sum_{t=0}^{T-1}
\sum_{k=1}^K \mathbb{E}\ [Y_k(t)] \leq \frac{B_2+V K}{\epsilon}  \label{eq:virbound2}
\end{align}
and
\begin{align}
 \liminf_{T\rightarrow \infty} \frac{1}{T} \sum_{t=0}^{T-1}
 \mathbb{E}[S_{tot}(t)] \geq S_{tot}^{*} -  \frac{B_2}{V}
 \label{eq:optbound2}
\end{align}
where $B_2= \frac{1}{2}\left( K ( (E_{max})^2  + (E_{max}^{tot})^2
\right)$ and $\epsilon$ can be defined as in \cite{Georgiadis:Resource06}.
\end{theorem}
\begin{proof}
The proof is similar to the proof of Theorem 1, and omitted for brevity. 
\end{proof}
Clearly, the expected $S_{tot}(t)$ can be pushed arbitrarily close to the optimum by choosing
$V$ sufficiently large. However, this leads to increasing bound
on the average virtual queue size and the convergence time of the algorithm  \cite{Karaca:cognitive}. 
%
%
%
%
%
%
%
%


\subsection{Overhead and Implementation Issues}
The  D-PPDU and EAD-PPDU algorithms described in
previous sections have some implementation issues and overhead which we explain here.

In practice, the 802.11 devices use virtual carrier sensing mechanism which is called Network Allocation Vector (NAV)  for collision avoidance. The NAV indicates how long the channel will be busy  so that the receiving stations remains silent until the current transmission ends. The NAV is sent within 802.11 MAC layer frame  headers. Since TF is also a MAC frame it should indicate the NAV as well. However, since the channel occupancy time (i.e., $T_s(t)$) can only be determined after receiving  buffer state information  it is not possible for the AP to indicate the duration of the whole transmission including the transmission duration of  BS, OT frames, and the PPDU duration within the TF. In order to solve this problem, we propose that the TF indicates the NAV only until the end  of the transmission of OT frame (see Fig.2). After the NAV expires the other user in the network can try to access to the channel which may collide with the transmission of OT frame.  In order to eliminate this issue, we must guarantee that the OT frame is transmitted first. Note that other users which wait  until the NAV expires have to wait DIFS (Distributed Coordination Function Inter Frame Space) then select a backoff timer. By using this fact, we propose  a similar procedure to the Hybrid coordination function Controlled Channel
Access (HCCA)  introduced in  802.11e \cite{ieeee}, in which OT frame is transmitted with PIFS (Point Coordination Function Inter Frame Space (IFS)) that is the IFS time for beacon packet transmission from the AP, and is less than  DIFS and without any backoff.  Then, the AP will access to the channel first to annoyance the optimal scheduling duration. 

F-PPDU and D-PPDU (also EAD-PPDU) algorithm introduce some overhead in terms of packet exchange and the protocol time, which we quantify here. Let $T_{TF}$, $T_{BS}$ and $T_{OT}$ be the transmission duration of a TF, BS  and OT frames, respectively.   We assume that these packets are transmitted with the basic rate (e.g., 6 Mbps for 802.11ac). Then, the total time required to transmit these frames is equal $T_{TF}+ T_{BS} + T_{OT}$.  In addition,  the implementation of D-PPDU  necessitates to add protocol time between the  frame exchanges. From Fig. 2,  the required protocol time is equal to $2 (xSIFS) + PIFS$. At a given time, the total required time with  D-PPDU is equal to $T_s^*(t) + T_{TF}+ T_{BS} + T_{OT}+  2 (xSIFS) + PIFS$  whereas with F-PPDU it is equal to $T_s + T_{TF}+ (xSIFS)$. As an example, if we only consider the total throughput maximization, for D-PPDU outperforms to F-PPDU at that time,  the following inequality must be satisfied:  $T_s - T_{min}(t) > xSIFS + PIFS + T_{BS} + T_{OT}$ (i.e., $T_s^*(t)=T_{min}(t)$). As an example,  xSIFS=16 $\mu$s and PIFS=25 $\mu$s with 802.11ac. By following the estimation in \cite{IEEEOFDMA}, transmitting a single user information requires roughly 2.6 $\mu$s. Then,  transmitting a TF which contains $y$ number of users information (i.e., station ID and the associated resource block) requires (56 + y $\times$ 2.6)  $\mu$s, where 56 $\mu$s is MAC and PHY  preambles overhead. We assume that BS and OT frames contains only one user information, hence $ T_{BS} = T_{OT}$=(56 + 1 $\times$ 2.6)=58.6$\mu$s.  
\begin{figure}[t!]
    \centering
    \includegraphics[width=0.8\columnwidth]{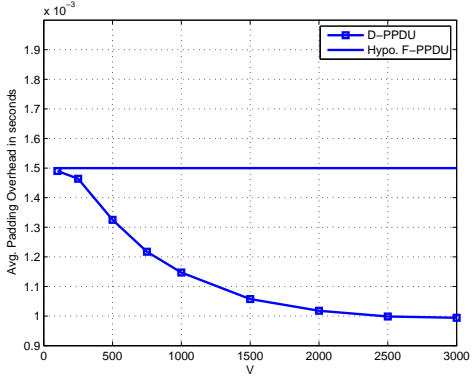}
    \caption{Avg. padding overhead with D-PPDU and Hypothetical F-PPDU.}
    \label{fig:fig3}
\end{figure}
\section{Numerical Results}
We consider a WLAN where there is an AP serving $N=100$ uplink users. We assume that there are $L=20$ identical groups in the network.   Specifically,  in each group there are $K=5$ users, and   the  frame  transmission  duration of each user has  a  Gamma  distribution \cite{Zhu:Gamma} with a different mean and variance. Particularly, the mean of transmission duration (i.e., $T_k(t)$) increases as the user index increases (i,e,. the first user has the  lowest mean whereas the user 5 has the highest mean). A round-robin scheduling policy is employed staring from the first group. In practice,  $T_s$ is usually not continuous, hence    in  the  simulation,  we  assume  that there is a discrete set of $T_s$ values available, which is given as $S_{T_s}=[0.05: 0.05: 12]$ milliseconds.  We set $C_n=0.65$ and  $P=25$ dBm (i.e, approximately 0.31 Watt) for all users. We assume that if a user cannot empty its buffer due to the short scheduling duration, the remaining data in the buffer is accumulated for the next scheduling time.
\begin{figure}[h]
    \centering
    \includegraphics[width=0.8\columnwidth]{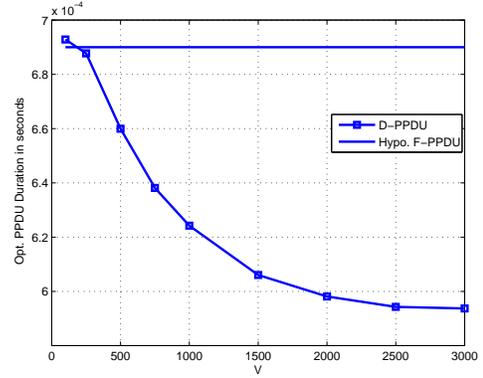}
    \caption{Opt. PPDU Duration with D-PPDU and Hypothetical F-PPDU.}
    \label{fig:fig4}
\end{figure}

First, we evaluate  the performance of F-PPDU algorithm as a solution the optimization problem \eqref{eq:problem}. Since F-PPDU uses a fixed $T_s$ first we find the optimal constant $T_s$. To do that we run a number of simulations for a given scenario with various values of $T_s$. Then, we pick the  $T_s$ value which has the minimum value and satisfy all the constraints in \eqref{eq:problem}.  We call this optimal value as Hypothetical value since it is impossible to find the optimal $T_s$ in advance.  Importantly,  it also means that  any other algorithms with a fixed and constant $T_s$ will have worse performance than the Hypo. F-PPDU in which $T_s$ is determined via an oracle (e.g.,  Hypo. F-PPDU has the highest performance). We observe that the total padding overhead sum over all the scheduled users and the optimal PPDU duration with Hypo. F-PPDU are found as $1.5$ and $0.69$ ms. 

Next, we evaluate the performance of D-PPDU with increasing values of $V$ since, according to Theorem 1, as $V$ increases the padding overhead should decrease. Figure~\ref{fig:fig3} depicts the total average padding overhead caused by Hypo. F-PPDU and D-PPDU algorithms. In our first simulation, we find that it is equal to 1.5 ms with Hypo. F-PPDU. As $V$ increases from 100 to 3000, it can be seen that the total average padding overhead decreases, and converges to around 1 ms, which means D-PPDU outperforms  Hypo- F-PPDU. The reason is that even though Hypo. F-PPDU satisfies all the constraints, and it is not opportunistic in the sense that it is not capable of exploiting the opportunities to further reduce the PPDU duration. Specifically, depending on the current situation of the constraints (i.e., the size of the virtual queues $X_n(t)$) D-PPDU may prefer not to allocate long PPDU duration even if the constraints may not be satisfied at that time instant. Figure~\ref{fig:fig4} shows the optimal PPDU duration $T_s^*$ with Hypo. F-PPDU and D-PPDU. It is equal to 0.69 ms in Hypo. F-PPDU. As $V$ increases $T_s^*$ decreases, and reaches to its optimal value which is approximately equal to 0.59 ms. From the energy point of view,  with fixed transmit power, it can be concluded that D-PPDU algorithm can achieve  $16\%$ energy gain compared to Hypo. F-PPDU algorithm. Since Hypo. F-PPDU cannot be achieved in practice, the gain will be higher compared to any other algorithm with a fixed $T_s$.

Figure~\ref{fig:fig5} depicts the average $F_k(t)$ values of the users in an arbitrary group. We recall that since all groups are identical it holds for other groups as well. When $V$ (e.g., $V \leq 500$) is small, the resulting $T_s$ is  high and all the constraints associated with these five users are satisfied at the point  higher than the minimum requirements (i.e., ${\mathbb E} [F_k(t)] > C_k$). However, with high values of $V$ the constraint of the  user 5 is satisfied as equality (i.e., ${\mathbb E}\left[F_5(t)\right] = 0.65$), and other users achieves a performance level higher than their minimum requirements. This is also intuitive in the sense that the optimal algorithm with the fairness consideration  aims to minimize the overhead of the user with the maximum PPDU duration. 
\begin{figure}[t]
    \centering
    \includegraphics[width=0.8\columnwidth]{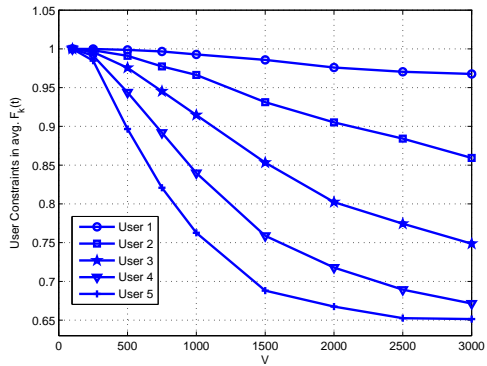}
    \caption{Avg. user constraints with F-PPDU vs. V.}
    \label{fig:fig5}
\end{figure}
\begin{figure}[t]
    \centering
    \includegraphics[width=0.8\columnwidth]{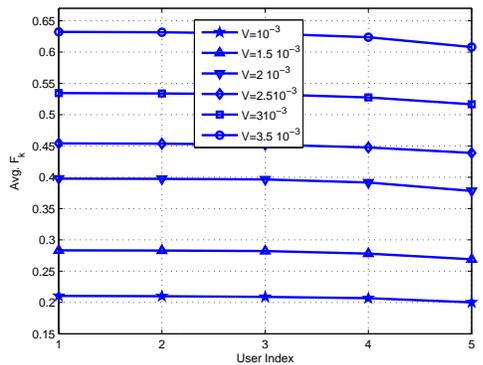}
    \caption{Avg. $F_k$ vs. V for various user index}
    \label{fig:fig6}
\end{figure}

Lastly, we evaluate the performance of EAD-PPDU algorithm. We use a different set of $V$ values  to facilitate the comparability of the two terms in \eqref{eq:policy2} since the unit of energy is on the order of $10^{-3}$, and the size of the virtual queues does not grow fast enough, which affects the convergence time significantly. We observe that the increase in the objective function in problem \eqref{eq:problem2}  with higher values of $V$ is small, however,  the convergence time is more than ten times higher compared to the current set of $V$ values.  Specifically, the number of iterations used with this current $V$ values is $10^5$, however, as $V$ increases (e.g. $V=100$) it requires more than $10^6$ iterations. Figure~\ref{fig:fig6} depicts the average number of times the users empty their buffers, and their energy constraints with respect to increasing values of $V$. The average number of time the  user 1 can empty its buffer is highest whereas the user 5 has the lowest one since  user 1 has the lowest transmission duration on average. Also, as $V$ increases,  all the users have more opportunity to empty their buffers as we verify in Theorem 2. We also observe that all constraints in
\eqref{eq:problem2} are satisfied. 

\section{Conclusion}
In this paper, we have developed  two resource allocation policies for the utilization of OFDMA for 802.11 networks, and have investigated their efficacy. We first highlighted the facts that fixed scheduling duration is necessary, and consequently the padding bits are unavoidable. Then,   by taking into account the padding overhead, energy and fairness issues, and using Lyapunov optimization technique  we have developed optimal algorithms 
by exploiting the time-varying transmission duration of users and realizing the benefits of transmission with dynamic scheduling duration. Numerical results show that our first policy minimizes the padding overhead in
addition to providing airtime fairness  required by the users. Our second policy can maximizes user's satisfaction in terms of airtime depending on their energy budget. Possible
extension of this paper includes investigating randomized algorithms with dynamic scheduling duration,  and reducing protocol overhead. 

\bibliographystyle{IEEEtran}
\bibliography{IEEEabrv,Wiopt}

\end{document}